\documentclass[a4paper,11pt,twoside,reqno,nonamelimits]{amsart}
\usepackage{a4wide}

\usepackage[usenames,dvipsnames]{pstricks}
\usepackage{mathptmx}
\usepackage{amsmath}
\usepackage{amssymb}
\usepackage{amsfonts}
\usepackage{amscd}
\usepackage{amsthm}
\usepackage{amsxtra}
\usepackage[all,pdf]{xy}  \CompileMatrices     
\usepackage{mathrsfs}
\usepackage{nicefrac}
\usepackage{bbold}
\usepackage{graphicx}
\usepackage{xcolor}
\usepackage{enumerate}
\usepackage{wrapfig}

\usepackage[norelsize,boxed,noend,linesnumbered]{algorithm2e}\DontPrintSemicolon

\theoremstyle{plain}
\newtheorem{theorem}{Theorem}
\newtheorem*{theorem*}{Theorem}

\newtheorem*{proposition*}{Proposition}

\newtheorem*{corollary*}{Corollary}

\newtheorem{lemma}[theorem]{Lemma}
\newtheorem*{lemma*}{Lemma}

\newtheorem*{observation*}{Observation}

\newtheorem*{conjecture*}{Conjecture}

\newtheorem*{question*}{Question}
\newtheorem{question}[theorem]{Question}
\newtheorem*{questions*}{Questions}

\newtheorem*{problem*}{Problem}

\newtheorem*{problems*}{Problems}

\theoremstyle{definition}

\newtheorem*{exercise*}{Exercise}

\theoremstyle{remark}

\newtheorem*{remark*}{Remark}

\newtheorem*{remarks*}{Remarks}

\newtheorem*{claim*}{Claim}

\newcommand{\subclass}[1]{}

\newcommand{\enumTi}[1]{\renewcommand{\theenumi}{#1}}

\newcommand{\alphenumi}{\enumTi{\alph{enumi}}}

\newcommand{\romenumi}{\enumTi{\roman{enumi}}}

\newlength{\hspaceforlengthglumpf}

\newcommand{\lt}{\left}
\newcommand{\rt}{\right}

\newcommand{\abs}[1]{{\lt\lvert{#1}\rt\rvert}}

\newcommand{\Babs}[1]{{\Bigl\lvert{#1}\Bigr\rvert}}

\newcommand{\sabs}[1]{{\lvert{#1}\rvert}}

\newcommand{\QQ}{\mathbb{Q}}
\newcommand{\RR}{\mathbb{R}}

\newlength{\algotabbingwidth}
\begin{document}
\author{Kaveh Khoshkhah$^*$}%


\title[On finding orientations]{\vspace*{-4ex}On finding orientations with fewest number of vertices with small out-degree}%

\begin{abstract}
  Given an undirected graph, each of the two end-vertices of an edge can ``own'' the edge.  Call a vertex ``poor'', if it owns at most one edge.  We give a polynomial time algorithm for the problem of finding an assignment of owners to the edges which minimizes the number of poor vertices.

  In the terminology of graph orientation, this means finding an orientation for the edges of a graph minimizing the number of edges with out-degree at most~1, and answers a question of Asahiro Jansson, Miyano, Ono (2014).
  \\[1ex]
  \textbf{Keywords: }Graph orientation, Graph algorithms.
\end{abstract}
\vspace*{-2ex}


\date{Tue Aug 12 13:16:33 EEST 2014.\\
$^*)$ Insitute of Computer Science, University of Tartu, J.~Liivi~2, 50409 Tartu, Estonia. \texttt{kaveh.khoshkhah@ut.ee}}

\maketitle

\newcommand{\MinLight}[1]{\textsc{Min-$#1$-Light}}


\section{Introduction}\label{sec:intro}
Let~$G$ be a simple\footnote{%
  Note that our main reference~\cite{AsahiroJanssonMiyanoOno13} uses multigraphs, but restricting to graphs is w.l.o.g.
} %
undirected graph.  An \textit{orientation} of~$G$ is a function $\Lambda$, which maps each undirected edge $\{u,v\}\in E(G)$ to one of the two possible directed edges $(u,v)$ or $(v,u)$.  We let $\Lambda(G)$ be the directed graph whose vertex set is $V(G)$ and whose set of (directed) edges is $\{ \Lambda(\{u,v\}) \mid \{u,v\}\in E(G) \}$.  For each $v\in V(G)$, denote by the \textit{out-degree of $u$ under~$\Lambda$} by
\begin{equation*}
  d^+_\Lambda(u) := \Babs{   \bigl\{ \{u,v\} \in E(G) \mid \Lambda(\{u,v\}) = (u,v) \bigr\}   }.
\end{equation*}
Fix an integer~$k\ge 0$.  A vertex $v\in V(G)$ is called \textit{$\Lambda$-$k$-light} (or just \text{$k$-light, light}) if $d^+_\Lambda(v) \le k$; otherwise it is called \textit{heavy}.  Asahiro et al.~\cite{AsahiroJanssonMiyanoOno12:isco,AsahiroJanssonMiyanoOno13} study the combinatorial optimization problem \MinLight{k} which asks for finding an orientation minimizing the number of $k$-light vertices.
For $k=1$, they exhibit classes of graphs on which the problem can be solved in polynomial time, and they ask the following open question.
\begin{question}[\cite{AsahiroJanssonMiyanoOno12:isco,AsahiroJanssonMiyanoOno13}]
  Is \MinLight{1} NP-hard for general graphs?
\end{question}

In this short note, we answer that question:

\begin{theorem}\label{thm:main-unweighted}
  \MinLight{1} on a graph with~$n_2$ vertice of degree at least~2, $n_1$ vertices of degree~1, and~$m$ edges can be solved by single maximum cardinality matching computation in a graph with $O(m)$ vertices and $O(m^2/n)$ edges.
\end{theorem}

Asahiro et al.~\cite{AsahiroJanssonMiyanoOno12:isco,AsahiroJanssonMiyanoOno13} mention a natural weighted version of the problem: the vertices have costs $c_v \in \QQ$, $v\in V(G)$, associated with them, and the objective is to find an orientation which minimizes the expression
$\sum_v c_v$
over all orientations~$\Lambda$, where the sum extends over all 1-light vertices~$v$.  Our result also gives the complexity of the weighted case.

\begin{theorem}\label{thm:main-weighted}
  For nonnegative weights, weighted \MinLight{1} on a graph with~$n$ vertices and~$m$ edges can be solved by single maximum weight matching computation in a graph with $O(m)$ vertices and $O(m^2/n)$ edges.
\end{theorem}
For weights which are not nonnegative, \MinLight{1} is NP-hard, since it includes as a special case (when all weights are $-1$) the problem \textsc{Min-$1$-Heavy}, for which Asahiro et al.~\cite{AsahiroJanssonMiyanoOno13} proved NP-hardness.

The proof of the theorems is in Section~\ref{sec:proof}.  Section~\ref{sec:concl} holds a conclusion.

\subsection*{Some notation}
We mostly adhere to standard notation.  Our (undirected) edges are 2-element subsets of the vertex set.  For a vertex~$v\in V(G)$, we denote by
$\displaystyle
  \delta(v) := \{ e \in E(G) \mid v \in e\}
$ 
the set of all edges incident on~$v$.  The degree of a vertex is denoted by $d(v) := \abs{\delta(v)}$.



\section{The algorithm for \MinLight{1}.}\label{sec:proof}
We first deal with the case that there are no vertices of degree~1.  For such a graph~$G$, construct a graph~$G'$ as follows.
Denote by $d(v)$ the degree of a vertex~$v$ in~$G$.
Start by letting~$G'$ be a copy of~$G$.  Then replace every edge $e=\{u,v\}$ by a path $u, u'_e, x_e, v'_e, v'$, by adding three new vertices $u'_e$, $x_e$, $v'_e$, and four new edges $\{u,u'_e\}$, $\{u'_e,x_e\}$, $\{x_e,v'_e\}$, $\{v'_e,v\}$.
We call the vertices $x_e$ \textit{connecting vertices}, and the edges $\{u'_e,x_e\}$ (and also $\{x_e,v'_e\}$) \textit{connecting edges,} and let $F_u := \{u'_e,x_e \mid e \in \delta(v)\}$.

Now, for each original vertex~$v$, do the following: replace~$v$ by $d(v)-2$ new vertices $v''_1,dots,v''_{d(v)-2}$.  Add $(d(v)-2)\cdot d(v)$ edges between the $v''_i$ and the $v'_e$, for every $i$ and every $e\in\delta(v)$.  Finally, choose two edges $e,f\in\delta(v)$ arbitrarily, and add an edge $g_v := \{v'_e, v'_f\}$, which we call the \textit{special edge.}

In this way, $G'$ contains pairwise disjoint ``gadgets'' ($\hat=$ induced subgraphs) $W_v$, $v\in V(G)$, each with $d(v)-2 + d(v)$ vertices and $(d(v)-2)\cdot d(v)$ edges.  If $uv\in E(G)$, then the gadgets $W_u$ and $W_v$ are joined to the connecting vertex $x_{uv}$ by $d(u)$, or $d(v)$, respectively, edges.   Cf.~Fig.~\ref{fig:two-gadgets}.  With $n := \abs{V(G)}$ and $m := \abs{E(G)}$, the resulting graph~$G'$ has
\begin{align*}
  m + \sum_{v\in V(G)} \bigl( d(v)-2 + d(v) \bigr)          &= 5m - 2n \text{ vertices, and} \\
  \sum_{v\in V(G)} \bigl( (d(v)-2)d(v) + 1 + d(v) \bigr) &\le \frac{4m^2}{n} + n - 2m \text{ edges.}
\end{align*}

\begin{figure}[ht]
  \begin{center}
  
 \includegraphics[height=60mm, width=100mm]{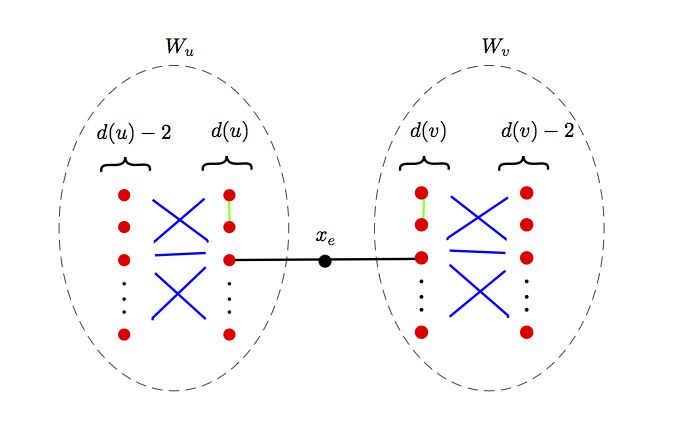}
  
 
  \end{center}
  \caption{Two ``gadgets'' $W_{u}$ and $W_{v}$ in the graph $G'$}\label{fig:two-gadgets}
\end{figure}

The following fact is crucial in the consturction.

\begin{lemma}\label{lem:matchings_in_Gprime}
  Let~$M$ be a maximal matching in~$G'$.
  For each $v\in V(G)$, there exists a matching $N_v$ which differs from~$M$ only on $E(W_v)$, and which satisfies either $N_v = M$ or $\abs{N_v\cap E(W_v)} = \abs{M\cap E(W_v)} +1$, and for which the following holds:
  With $k := \abs{M \cap F_v}$, we have
  \begin{equation}\label{eq:max_match_eq}
    \abs{ N_v\cap \bigl( E(W_v)\cup F_v \bigr) } =
    \begin{cases}
      d(v)-1, &\text{ if $0\le k \le 1$}\\
      d(v),   &\text{ if $k\ge 2$.}\\
    \end{cases}
  \end{equation}
\end{lemma}
\begin{proof}
  Let~$M$ be such a maximal matching.
  If $k=0$ and the special edge~$g_v$ is not in~$M$, then, to obtain~$N_v$, we replace $M\cap E(W_v)$ by the edges of a perfect matching of~$W_v$, which consists of $g_v$ plus a perfect bipartite matching between the $v''_i$ and the $v'_e$.  This increases the number of edges in the matching by~1.
  If $k\ge 2$ and $g_v \in M$, then at least two of the vertices $v''_i$, $i=1,\dots,d(v)-2$ are exposed.  To obtain~$N_v$, we delete~$g_v$ from $M$ and add two edges from exposed vertices in $v''_i$, $i=1,\dots,d(v)-2$, to the end-vertices of~$g_v$, thus increasing the number $W_v$-edges in the matching by~1.

  In all other cases, we leave~$M$ unchanged: $N_v := M$.

  The equations~\eqref{eq:max_match_eq} can now be easily derived.
  If $k=0$, then $N_v\cap E(W_v)$ is a perfect matching in~$W_v$, of size $d(v)-1$.
  If $k=1$, then any maximal matching leaves either one or one of the vertices $v''_i$, $i=1,\dots,d(v)-2$, unmatched.  If one is left unmatched, then the matching must contain the special edge~$g_v$, so $\abs{M\cap E(W_v)}=d(v)-2$, implying~\eqref{eq:max_match_eq}.
  If $k\ge 2$, then taking into account that $g_v\not\in N_v$, equation~\eqref{eq:max_match_eq} readily follows.
\end{proof}


We can now prove that solving the maximum (cardinality) matching problem on~$G'$ is equivalent to solving \MinLight{1} on~$G$.

\begin{lemma}\label{lem:main}
  If~$G$ has no vertices of degree~1, then \MinLight{1} on~$G$ can be solved by computing a maximum matching in~$G'$.
\end{lemma}
\begin{proof}
  Firstly, consider an orientation $\Lambda$ of~$G$.  We will construct a matching~$M = M(\Lambda)$ in~$G'$ with the property that, for all $v\in V(G)$,
  \begin{subequations}\label{eq:matching-orientation-relation}
    \begin{equation}
      d(v) - M \cap \bigl( E(W_v)\cup F_v \bigr) =
      \begin{cases}
        1, &\text{ if $d^+_\Lambda(v) \le 1$, and}\\
        0, &\text{ if $d^+_\Lambda(v) \ge 2$,}
      \end{cases}
    \end{equation}
    and so
    \begin{equation}
      \abs{   \bigl\{ v\in V(G) \bigm| d^+_\Lambda(v) \le 1 \bigr\}    } = 2m - \abs{M}.
    \end{equation}
  \end{subequations}
  For every directed edge $(u,v)$ in $\Lambda(G)$, choose the edge $\{u'_e,x_e\}$ to be in~$M$.  This means that, for every~$v\in V(G)$, we have
  \begin{equation}\label{eq:match-orient-local}\tag{$*$}
    \abs{M\cap F_v} = d^+_\Lambda(v).
  \end{equation}
  Then extend~$M$ arbitrarily to a maximal matching by adding edges from the $E(W_v)$, $v\in V(G)$.  Note that~$M$ is unchanged on the sets~$F_v$, $v\in V(G)$, so that~\eqref{eq:match-orient-local} still holds.  Finally, for each $v\in V(G)$, apply Lemma~\ref{lem:matchings_in_Gprime}, and replace the edges in ~$M\cap E(W_v)$, by the edges of $N_v \cap E(W_v)$.  The result is a matching satisfying~\eqref{eq:matching-orientation-relation}.


  Secondly, let~$M$ be a maximum matching in $G'$.  We will construct an orientation $\Lambda =\Lambda(M)$ of~$G$ satisfying~\eqref{eq:matching-orientation-relation}.
  For each $\{u,v\} \in E(G)$, if $\{u'_e,x_e\} \in M$, let $\Lambda(\{u,v\}) := (u,v)$; if $\{v'_e,x_e\} \in M$, let $\Lambda(\{u,v\}) := (v,u)$.  If the vertex~$x_e$ is $M$-exposed, then chose one of $(u,v)$, $(v,u)$ arbitrarily for $\Lambda(\{u,v\})$.

  In view of Lemma~\ref{lem:matchings_in_Gprime}, $M$ must coincide with each of the $N_v$, and hence the equations~\eqref{eq:max_match_eq} hold.  But, by the construction of~$\Lambda$, for each $v\in V(G)$,
  \begin{equation}\label{eq:match-orient-local}\tag{$*$}
    \abs{M\cap F_v} \le d^+_\Lambda(v).
  \end{equation}
  Hence, we conclude that
  \begin{equation*}
    \Babs{   \bigl\{ v\in V(G) \bigm| d^+_\Lambda(v) \le 1 \bigr\}   }
    \le
    \Babs{   \bigl\{ v\in V(G) \bigm| \sabs{M\cap F_v} \le 2  \bigr\}    }
    =
    2m - \abs{M(\Lambda)}.
  \end{equation*}

  We conclude.  Denoting by $\pi$ the smallest number of light vertices in any orientation of~$G$, and by $\mu$ the largest cardinality of a matching in~$G'$, we have
  \begin{equation*}
    \pi
    \le
    \Babs{   \bigl\{ v\in V(G) \bigm| d^+_{\Lambda(M)}(v) \le 1 \bigr\}   }
    =
    2m - \mu
    \le
    2m - \abs{M(\Lambda)}
    \le
    \pi,
  \end{equation*}
  which concludes the proof of the lemma.
\end{proof}

We can now prove Theorem~\ref{thm:main-unweighted}.

\begin{proof}[Proof of Theorem~\ref{thm:main-unweighted}.]
  To get rid of vertices of degree~1 in the input graph, for each such vertex~$v$, add three more vertices $v_1,v_2$, and four edges $\{v,v_1\}$, $\{v_1,v_2\}$, $\{v_2,v_3\}$, $\{v_3,v\}$.  In other words, we replace each degree-1 vertex by a 4-cycle.  A 4-cycle can have~2 heavy vertices, opposite each other, and the other~2 vertices will be light; the edge leaving the cycle will not change that.  From this, it can be readily verified that \MinLight{1} on the original graph is equivalent to \MinLight{1} on the modified graph.  Lemma~\ref{lem:main} now yields the result.
\end{proof}

\subsection*{The weighted case}
The weighted case differs only in technical aspects.

\begin{proof}[Sketch of the proof of Theorem~\ref{thm:main-weighted}.]
  First of all, note that the degree-1 vertices can be taken care of in just the same way as in the non-weighted case: just give the new vertices a cost of~0.  Then, walking through the proof of Lemma~\ref{lem:main}, we see that the argument is still valid for weighted matchings and costs punishing the light vertices. Indeed, that's the reason why we phrased Lemma~\ref{lem:matchings_in_Gprime} in the way we did: if~$c_v$ is the cost incurred if vertex~$v$ is light, give each edge in $E(W_v)\cup F_v \subset E(G')$ a weight of $c_v$. Then, as in the proof of Lemma~\ref{lem:matchings_in_Gprime}, denoting by $\pi$ the cost incurred by the light vertices in any orientation of~$G$, and by $\mu$ the largest weight of a matching in~$G'$, and with $Q := \sum_{e=\{u,v\}\in E(G)}(c_u+c_v)$, it's easy to show that
  \begin{equation*}
    \pi
    \le
    \Babs{   \bigl\{ v\in V(G) \bigm| d^+_{\Lambda(M)}(v) \le 1 \bigr\}   }
    =
    Q - \mu
    \le
    Q - \abs{M(\Lambda)}
    \le
    \pi,
  \end{equation*}
  which concludes the proof of the theorem.
\end{proof}



\section{Conclusion}\label{sec:concl}
Seeing as weighted \MinLight{1} can be solved in polynomial time by matching techniques for nonnegative weights, it is natural to ask for a description by linear inequalities of the polyhedron $P_G \subset \RR^{V(G)}$ defined by the problem: $P_G$ is the dominant (see~\cite{SchrijverBk03} for details) of the convex hull of the points $x(\Lambda) \in \RR^{V(G)}$, which have $x(\Lambda)_v = 1$ if $v$ is $\Lambda$-poor, and $x(\Lambda)_v = 0$ otherwise.

Kyncl et al.~\cite{KyncLidickyVyskocil09} study the so-called minimum irreversible $k$-conversion problem, which is closely related to \MinLight{*}.  In fact, the only difference between  \MinLight{k} and
Minimum Irreversible $(k+1)$-Conversion is that the latter requires the orientations to be acyclic.  Kyncl et al.\ prove that Minimum Irreversible $2$-conversion is NP-hard, even for graphs of maximum degree~4, but for 3-regular graphs, it is equivalent to finding a vertex feedback set (which can be done in poly-time~\cite{UenoKG88}).

Since the complexity of Minimum Irreversible $2$-Conversion is open for subcubic graphs, in the light of our result, we conjecture that there might be a matching-based algorithm for that problem.

\medskip\noindent
I would like to thank Dirk Oliver Theis for his support, inspiring discussions and useful comments.


\bibliographystyle{amsplain}
\bibliography{orientation}

\providecommand{\bysame}{\leavevmode\hbox to3em{\hrulefill}\thinspace}
\providecommand{\MR}{\relax\ifhmode\unskip\space\fi MR }
\providecommand{\MRhref}[2]{%
  \href{http://www.ams.org/mathscinet-getitem?mr=#1}{#2}
}
\providecommand{\href}[2]{#2}
\begin{thebibliography}{1}

\bibitem{AsahiroJanssonMiyanoOno12:isco}
Yuichi Asahiro, Jesper Jansson, Eiji Miyano, and Hirotaka Ono, \emph{Graph
  orientations optimizing the number of light or heavy vertices}, ISCO, 2012,
  pp.~332--343.

\bibitem{AsahiroJanssonMiyanoOno13}
\bysame, \emph{Degree-constrained graph orientation: Maximum satisfaction and
  minimum violation}, WAOA, 2013, pp.~24--36.

\bibitem{KyncLidickyVyskocil09}
Jan Kyncl, Bernard Lidick{\`y}, and Tom{\'a}{\v{s}} Vyskocil,
  \emph{Irreversible 2-conversion set is np-complete}, KAM-DIMATIA Series 933,
  2009.

\bibitem{SchrijverBk03}
Alexander Schrijver, \emph{Combinatorial optimization. {P}olyhedra and
  efficiency.}, Algorithms and Combinatorics, vol.~24, Springer-Verlag, Berlin,
  2003. \MR{1956926}

\bibitem{UenoKG88}
Shuichi Ueno, Yoji Kajitani, and Shin'ya Gotoh, \emph{On the nonseparating
  independent set problem and feedback set problem for graphs with no vertex
  degree exceeding three}, Discrete Mathematics \textbf{72} (1988), no.~1-3,
  355--360.

\end{thebibliography}
\end{document}